\documentclass{amsart}
\usepackage{hyperref}
\usepackage{amsmath,amsthm,amscd,amssymb}
\usepackage{enumerate}

\newcommand{\arxiv}[1]{\href{http://arxiv.org/#1}{arXiv:#1}}
\newcommand*{\mailto}[1]{\href{mailto:#1}{\nolinkurl{#1}}}

\newtheorem{theorem}{Theorem}[section]
\newtheorem{lemma}[theorem]{Lemma}
\newtheorem{corollary}[theorem]{Corollary}

\numberwithin{equation}{section}

\newcommand{\C}{\mathbb{C}}
\newcommand{\R}{\mathbb{R}}

\newcommand{\E}{\mathrm{e}}
\newcommand{\I}{\mathrm{i}}

\newcommand{\siul}{\sigma^{\mathrm{u,l}}}
\newcommand{\siu}{\sigma^{\mathrm{u}}}
\newcommand{\sil}{\sigma^{\mathrm{l}}}
\newcommand{\lau}{\lambda^{\mathrm{u}}}
\newcommand{\lal}{\lambda^{\mathrm{l}}}

\newcommand{\Res}{\mathop{\rm Res}}

\renewcommand{\Im}{\mathop{\rm Im}}

\newcommand{\beq}{\begin{equation}}
\newcommand{\eeq}{\end{equation}}
\newcommand{\bal}{\begin{align}}
\newcommand{\eal}{\end{align}}
\newcommand{\nn}{\nonumber}
\newcommand{\al}{\alpha}

\newcommand{\om}{\omega}
\newcommand{\si}{\sigma}
\newcommand{\pa}{\partial}

\newcommand{\la}{\lambda}
\newcommand{\ov}{\overline}

\numberwithin{equation}{section}

\begin{document}

\title[{A Paley--Wiener Theorem for Periodic Scattering}]{A Paley--Wiener Theorem
for Periodic Scattering with Applications to the Korteweg--de Vries Equation}

\author[I. Egorova]{Iryna Egorova}
\address{B.Verkin Institute for Low Temperature Physics\\
47 Lenin Avenue\\61103 Kharkiv\\Ukraine}
\email{\mailto{iraegorova@gmail.com}}

\author[G. Teschl]{Gerald Teschl}
\address{Faculty of Mathematics\\
Nordbergstrasse 15\\ 1090 Wien\\ Austria\\ and\\ International Erwin Schr\"odinger
Institute for Mathematical Physics\\ Boltzmanngasse 9\\ 1090 Wien\\ Austria}
\email{\mailto{Gerald.Teschl@univie.ac.at}}
\urladdr{\url{http://www.mat.univie.ac.at/~gerald/}}

\thanks{Research supported by the Austrian Science Fund (FWF) under Grant No.\ Y330.}
\thanks{Zh. Mat. Fiz. Anal. Geom. {\bf 6:1}, 21--33 (2010)}


\begin{abstract}
Consider a one-dimensional Schr\"odinger operator which is a short-range perturbation
of a quasi-periodic, finite-gap operator. We give necessary and sufficient conditions on the left, right reflection
coefficient such that the difference of the potentials has finite support to the left, right, respectively.
Moreover, we apply these results to show a unique continuation type result for solutions
of the  Korteweg--de Vries equation in this context. By virtue of the Miura transform an
analogous result for the modified Korteweg--de Vries equation is also obtained.
\end{abstract}

\maketitle

\section{Introduction}

Since the seminal work of Gardner et al.\ \cite{GGKM} in 1967 the
inverse scattering transform is one of the main tools for solving
the Korteweg--de Vries (KdV) equation \beq \label{KdV}q_t(x,t) =
-q_{xxx}(x,t) + 6 q(x,t) q_x(x,t). \eeq Since it very much
resemblances the use of the classical Fourier transform method to
solve linear partial differential equations, the inverse scattering
transform is also known as the nonlinear Fourier transform.
Moreover, the linear and nonlinear Fourier transform share many
other properties one of which, namely the Paley--Wiener theorem,
will be the main subject of this paper.

Let $L_q= -\frac{d^2}{dx^2}+q(x)$ be the one-dimensional
Schr\"odinger operator. Assume that $q(x)$ decays sufficiently fast
such that one can associate left/right reflection coefficients
$R_\pm(\la)$ with it. In their seminal paper Deift and Trubowitz
\cite{dt} observed that if $L_q$ has no eigenvalues, then $q(x)$ has
support in $(-\infty,a)$ if $R_+(\la)$ has an analytic extension
satisfying the growth condition $\sqrt\la\,R_+(\la)
=O(\E^{-2a\I\sqrt{\la}})$. Combining this result with some Hardy
space theory enabled Zhang \cite{zh} to prove unique continuation
results for the KdV equation. To be able to use the result from
Deift and Trubowitz, commutation methods (see \cite{dt}, \cite{gs}, \cite{gt})
were used to remove all eigenvalues. If one wants to avoid this extra
step, this raises the question what is needed in addition to the growth
condition on $R_+(\la)$ in the case when eigenvalues are present. It
seems that Aktosun \cite{A} was the first to realize that there is
an extra condition on the residue of $R_+(\la)$ at an eigenvalue.
However, it seems he did not notice that this condition, together
with the growth estimate, is also sufficient. This Paley--Wiener
type theorem will be our fist main result, Theorem~\ref{thm:pw}. In
fact, we will establish the result for more general case of
potentials which are asymptotically close to a real-valued,
quasi-periodic, finite-gap potential. We then apply this to
solutions of the KdV equation and prove a unique continuation result
(Theorem~\ref{thm:kdv2}) for the KdV equation in this setting. Again
we extend the results from \cite{zh} to solutions which are not
decaying but rather are asymptotically close to some quasi-periodic,
finite-gap solution $p(x,t)$. While these results  are only special
cases of some more general results which can be proven using modern
harmonic analysis (see for example \cite{ekpv} and the references
therein), we still present them here since the proof is much simpler
and does not require advanced harmonic analysis (note that in the
discrete case an even simpler argument is possible \cite{kt}).

For further results on the Cauchy problem of the KdV equation with initial
conditions supported on a half-line see Rybkin \cite{ry} (cf.\ also Tarama \cite{ta})
and the references therein.

\section{Some general facts on quasi-periodic, finite-gap potentials}
\label{secfgp}

In this section we briefly recall some basic facts on finite gap
potentials needed later one. For further information we refer to,
for example, \cite{GH}, \cite{GRT}, \cite{M}, or \cite{NMPZ}.

Let $L_p$ be  a one-dimensional Schr\"odinger operator with a finite
gap potential $p(x)$ associated with the hyperelliptic Riemann
surface of the square root $Y(\la)^{1/2}$, where
\[
 Y(\la)=-\prod_{j=0}^{2r} (\la-E_j),\quad E_0 < E_1 <
  \dots < E_{2r}.
\]
The spectrum of $L_p$ consists of $r+1$ bands:
\[
\sigma = \sigma(L_p) = [E_0, E_1]\cup\dots\cup[E_{2j-2},E_{2j-1}]\cup\dots\cup[E_{2r},\infty)
\]
and the potential $p(x)$ is uniquely determined by its associated
Dirichlet divisor
\[
 \left\{(\mu_1,\si_1), \dots,(\mu_r, \si_r)\right\},
\]
where $\mu_j \in [ E_{2j-1}, E_{2j}]$ and $\si_j \in \{+1,-1\}$.

We denote by $ \psi_\pm(\la,x)$ the corresponding Weyl solutions of
$L_p\psi_\pm=\la\psi_\pm$, normalized according to
$\psi_\pm(\la,0)=1$ and satisfying $\psi_\pm(\la,.)\in
L^2((0,\pm\infty))$ for $\la\in\mathbb{C}\setminus\si$. These
functions  are meromorphic for $\la\in\C\backslash\si$ with
continuous limits (away from its singularities described below) on
$\si$ from the upper and lower half plane. Unless otherwise stated
we will always chose the limit from the upper half plane (the one
from the lower half plane producing just the corresponding complex
conjugate number).

When there is the need to distinguish between these limits we will cut the complex plane along
the spectrum $\si$ and denote the upper and lower sides of the cuts by $\siu$ and
$\sil$. The corresponding points on these cuts will be denoted by
$\lau$ and $\lal$, respectively. Moreover, we will write
\[
f(\lau) := \lim_{\varepsilon\downarrow0} f(\la+\I\varepsilon),
\qquad f(\lal) := \lim_{\varepsilon\downarrow0}
f(\la-\I\varepsilon), \qquad \la\in\si.
\]
Let $m_\pm(\la)=\frac{\pa}{\pa x}\psi_\pm(\la,0)$ be the Weyl
functions of operator $L_p$. Due to our normalization, for every
Dirichlet eigenvalue $\mu_j$ the Weyl functions might have poles. If
$\mu_j$ is in the interior of its gap, precisely one Weyl function
$m_+(\la)$ or $m_-(\la)$ will have a simple pole. Otherwise, if
$\mu_j$ sits at an edge, both will have a square root singularity.
Hence we divide the set of poles accordingly:
\begin{align*}
M_+ &=\{ \mu_j\mid\mu_j \in (E_{2j-1},E_{2j}) \text{ and } m_+(\la) \text{ has a simple pole}\},\\
M_- &=\{ \mu_j | 1 \le j \le r\} \backslash M_+.
\end{align*}
In addition, we set \beq\label{S2.6} \delta_\pm(z) :=
\prod_{\mu_j\in M_\pm}(z-\mu_j), \quad \tilde{\psi}_\pm(\la,x) :=
\delta_\pm(\la) \psi_\pm(\la,x)\eeq such that $ \tilde{\psi}_\pm$
are analytic for $\la\in\C\backslash\si$.
 Note
that we have chosen $M_-$ such that
\beq\label{1.881}\delta_-(\la)\delta_+(\la) = \prod_{j=1}^r
(\la-\mu_j).\eeq  Finally, introduce the function
\begin{equation}\label{1.88}
g(\la)= -\frac{\prod_{j=1}^r(\la - \mu_j)}{2 Y^{1/2}(\la)} = \frac{1}{W(\psi_+(\la),\psi_-(\la))},
\eeq
where the branch of the square root is chosen such that
\[
\frac{1}{\I} g(\lau) = \Im(g(\lau))  >0 \quad
\mbox{for}\quad \lambda\in\si,
\]
where $W(f,g)= f(x) g'(x) - f'(x) g(x)$ is the usual Wronski determinant.

Recall also the well-known asymptotics
\beq\label{estg}
g(\la) = \frac{\I}{2\sqrt{\la}} + O(\la^{-1})
\eeq
and
\beq\label{asympsi}
\psi_\pm(\la,x,t)=\E^{\pm\I\sqrt{\la} x} \left(1 +
O\Big(\frac{1}{\sqrt{\la}}\Big) \right),
\eeq
as $\la\to\infty$.

\section{Scattering theory in a nutshell}

In this section we give a brief review of scattering theory with respect to
quasi-periodic, finite-gap backgrounds. We refer to \cite{BET} for further details and proofs
(see also \cite{F1}, \cite{F2}, \cite{F3}, \cite{MT}).

Let $L_p$ be a Schr\"odinger operators with a real-valued,
quasi-periodic, finite-gap potentials $p(x)$ as in the previous
section. Let $q(x)$ be a real-valued function satisfying
\beq\label{S.2}
\int_\R (1+|x|^2)|q(x) - p(x)| dx < \infty
\eeq
and let
\[
L_q :=- \frac{d^2}{dx^2} +q(x),\quad x\in \R,
\]
be the ``perturbed" operator. The spectrum of $L_q$ consists of a
purely absolutely continuous part $\sigma$ plus a finite number of
eigenvalues situated in the gaps, \[ \si^d:=
\{\lambda_1,\dots,\lambda_s\}\subset\R\setminus\sigma. \]

The Jost solutions of the equation
\[
\left(-\frac{d^2}{dx^2}+q(x)\right)\phi(x)= \la \phi(x),\quad \la\in \C,
\]
that are asymptotically close to the Weyl solutions of the background
operators as $x\to\pm\infty$
can be represented with the help of the transformation operators as
\begin{equation}\label{S2.2}
\phi_\pm(\la,x) =\psi_\pm(\la,x)\pm\int_{x}^{\pm\infty}
 K_\pm(x,y)\psi_\pm(\la,y) dy,
\end{equation}
where $K_\pm(x,y)$ are real-valued functions satisfying
\begin{equation}\label{A.5}
 K_\pm(x,x)=\pm\frac{1}{2}\int_x^{\pm\infty} (q(y)-p(y))dy.
\end{equation}
\beq\label{A.6} |K_\pm(x,y)|\leq C(x_0)
\int_{\frac{x+y}{2}}^{\pm\infty} |q(z)-p(z)|d z,\quad \pm y>\pm
x>\pm x_0.\eeq Representation \eqref{S2.2} shows, that the Jost
solutions inherit all singularities of the background Weyl solutions
as well as the asymptotics
\beq\label{asymphi}\phi_\pm(\la,x,t)=\E^{\pm \I\sqrt{\la} x} \left(1
+ O\Big(\frac{1}{\sqrt{\la}}\Big) \right), \quad \la\to\infty. \eeq
Hence we set (recall \eqref{S2.6})
\[
\tilde\phi_\pm(\la,x)=\delta_\pm(\la) \phi_\pm(\la,x)
\]
such that the functions $\tilde\phi_\pm(\la,x)$ have no poles in the interior of
the gaps of $\si$. For every eigenvalue we can then
introduce the corresponding norming constants
\[
\left(\gamma_k^\pm\right)^{-1}=\int_\R \tilde\phi_\pm^2(\la_k,x) dx.
\]
Since at every eigenvalue the two Jost solutions must be linearly
dependent, we have \beq \label{ck}\tilde\phi_+(\la_k,x) = c_k
\tilde\phi_-(\la_k,x). \eeq  Furthermore, introduce the scattering
relations
\beq
T(\lambda) \phi_\pm(\lambda,x)
=\overline{\phi_\mp(\lambda,x)} + R_\mp(\lambda)\phi_\mp(\lambda,x),
\quad\lambda\in\siul,
\eeq
where the transmission and reflection
coefficients are defined as usual,
\beq\label{2.17}
T(\lambda):= \frac{W(\overline{\phi_\pm(\la)}, \phi_\pm(\la))}
{W(\phi_\mp(\la),\phi_\pm(\la))}, \quad
R_\pm(\la):= - \frac{W(\phi_\mp(\la),\overline{\phi_\pm(\la)})}
{W(\phi_\mp(\la), \phi_\pm(\la))}, \quad\la\in \siul.
\eeq
Since
\[
T(\lambda)= \frac{W(\psi_+(\la),\psi_-(\la))}{W(\phi_+(\la),\phi_-(\la))} = \frac{1}{g(\la) W(\phi_+(\la),\phi_-(\la))}
\]
the transmission coefficient has a meromorphic extension to the set
$\C\backslash\si$ with simple poles at the eigenvalues $\la_k$ and
residues given by (cf. \cite{BET})\beq\label{residues}
\Res_{\la=\la_k} T(\la) = 2 Y^{1/2}(\la_k) c_k^{\pm 1}\gamma_k^\pm.
\eeq It is important to emphasize that the reflection coefficients
in general do not have a meromorphic extension.

The sets
\begin{align}\nn
\mathcal{S}_\pm(q) := \Big\{ & R_\pm(\la),\; \la\in\si; \:
\lambda_1,\dots,\lambda_s\in\R\setminus \sigma,\;
\gamma_1^\pm,\dots,\gamma_s^\pm\in\R_+\Big\}
\end{align}
are called the right, left scattering data, respectively. Given $p(x)$, the
potential $q(x)$ can be uniquely recovered from each one of them as follows:

The kernels $K_\pm(x,y)$ of the transformation operators satisfy
the Gelfand-Levitan-Marchenko (GLM) equations
\begin{equation}\label{ME}
K_\pm(x,y) + F_\pm(x,y) \pm \int_x^{\pm\infty} K_\pm(x,z)
F_\pm(z,y)d z =0, \quad \pm y>\pm x,
\end{equation}
where \footnote{Here we have used the notation
$\oint_\sigma f(\lambda)d\la := \int_{\siu} f(\lambda)d\la -
\int_{\sil} f(\lambda)d\la$.}
\begin{align}\label{4.2}
F_\pm(x,y) &= \frac{1}{2\pi\I}\oint_\sigma R_\pm(\lambda)
\psi_\pm(\lambda,x) \psi_\pm(\lambda,y) g(\lambda)d\la\\ \nn &\quad
+ \sum_{k=1}^s \gamma_k^\pm \tilde\psi_\pm(\lambda_k,x)
\tilde\psi_\pm(\lambda_k,y).
\end{align}
Conversely, given $\mathcal{S}_\pm(q)$ we can compute $F_\pm(x,y)$
and solve \eqref{ME} for $K_\pm(x,y)$. The potential $q(x)$ can then
be recovered from \eqref{A.5}.

\section{Perturbations with finite support on one side and the nonlinear Paley--Wiener theorem}

In this section we want to look at the special case where $q(x)$ will be equal to $p(x)$ for
$x\leq a$ or $x\ge b$. Our main result in this section is the following theorem:

\begin{theorem}[Nonlinear Paley--Wiener]\label{thm:pw}
Suppose $q(x)$ satisfies \eqref{S.2}. Then we have $q(x)=p(x)$ for
$x\le a$ if and only if $\delta_+(\la)R_-(\la)$ has an analytic
extension to $\C\backslash(\si\cup\si^d)$ such that
\begin{align}\label{cond1}
&\Res_{\la=\la_k} \frac{g(\la)}{\delta_-(\la)^2} R_-(\la) = \gamma^-_k,\\\label{cond12}
&\sqrt \la R_-(\la)= O(\E^{2a \I\sqrt{\la}}) \quad \text{as}\  \la\to\infty.
\end{align}
Similarly, we have $q(x)=p(x)$ for $x\ge b$ if and only if
$\delta_-(\la)R_+(\la)$ has an analytic extension to
$\C\backslash(\si\cup\si^d)$ such that
\begin{align}\label{cond2}
& \Res_{\la=\la_k} \frac{g(\la)}{\delta_+(\la)^2} R_+(\la)= \gamma^+_k,\\\label{cond22}
&\sqrt\la R_+(\la)= O(\E^{-2\I b
\sqrt{\la}}) \quad \text{as}\ \la\to\infty.
\end{align}
\end{theorem}

\begin{proof}
Suppose first that $q(x)=p(x)$ for $x\le a$. Then we have
\[
\phi_+(\la,x) =  \alpha(\la) \psi_+(\la,x) + \beta(\la) \psi_-(\la,x), \qquad x \leq a,
\]
and thus
\begin{align*}
\alpha(\la) &= \frac{W(\psi_-(\la),\phi_+(\la))}{W(\psi_-(\la),\psi_+(\la))}
= -g(\la) W(\psi_-(\la),\phi_+(\la)),\\
\beta(\la) &= -\frac{W(\psi_+(\la),\phi_+(\la))}{W(\psi_-(\la),\psi_+(\la))}
= g(\la) W(\psi_+(\la),\phi_+(\la)),
\end{align*}
where the Wronskians can be evaluated at any $x\leq a$. In
particular, $\alpha(\la)$ is analytic in $\C\backslash\si$ and
$\beta(\la)$ is meromorphic in $\C\backslash\si$ with the only
simple poles at $\la\in M_+$. Note also that $\beta(\la)$ has simple
zeros at $\la\in M_-$and thus
\beq\label{betatilde}
\tilde{\beta}(\la)= \frac{\delta_+(\la)}{\delta_-(\la)} \beta(\la)
\eeq
is analytic in $\C\backslash\si$. Hence, since $\alpha(\la)$
vanishes at each eigenvalue $\la_k$, evaluating
\[
\tilde\phi_+(\la,x) =  \alpha(\la) \tilde\psi_+(\la,x) + \tilde\beta(\la) \tilde\psi_-(\la,x),
\qquad x \leq a,
\]
at $\la_k$ shows $\tilde{\beta}(\la_k) = c_k$ and formula
\eqref{cond1} follows from \eqref{1.881}, \eqref{1.88},
\eqref{residues}, \eqref{betatilde} and
\[
R_-(\la) = \frac{\beta(\la)}{\al(\la)},
\]
respectively,
\[
g(\la)R_-(\la)\delta_-^{-2}
(\la)=\tilde\beta(\la)T(\la)(2Y^{1/2}(\la))^{-1}
\]
The asymptotic behavior \eqref{cond12} follows using the well-known asymptotical
formula $\al(\la)= T(\la)^{-1} = 1+o(1)$, \eqref{estg}, \eqref{asympsi}, and \eqref{asymphi}.
This finishes the first part.

To see the converse, note that the growth estimate implies that we
can evaluate the integral in \eqref{4.2} by the residue theorem by
using a large circular arc of radius $r$ whose contribution will
vanish as $r\to\infty$ by the Jordan Lemma. Hence the integral in
\eqref{4.2} is just the sum over the residues which are precisely at the
eigenvalues $\la_k$ and by our conditions \eqref{cond1}
on the poles of integrand   it will cancel with the other sum in
\eqref{4.2}. Thus $F(x,y)=0$ for $y<x<a$ which by the GLM equation
implies $K_-(x,y)=0$ for $y<x<a$ which finally implies $p(x)-q(x)=0$
for $y<x<a$.
\end{proof}

As a consequence note that the scattering data $\mathcal{S}_\pm(q)$ are
determined by $R_\pm(\la)$ alone in such a situation since the eigenvalues and
norming constants can be read off from the poles of $R_\pm(\la)$. In
particular, combining this result with the results from  \cite{BET}
we can give the following characterization of scattering data
which give rise to potential supported on a half line.

Let
\beq
\om_{z z^*} = \left( \frac{Y^{1/2}(z)}{\la - z} + P_{z z^*}(\la) \right)
\frac{d\la}{Y^{1/2}(\la)}
\eeq
be the normalized Abelian differential of the third
kind with poles at $z$ and $z^*$ on the Riemann surface associated
with the function $Y^{1/2}(\la)$. Here $P_{z z^*}(z)$ is a polynomial of
degree $g-1$ which has to be chosen such that $\om_{z z^*}$ has vanishing
$a$-periods (the $a$-cycles are chosen to surround
the gaps of the spectra, changing sheets twice). Furthermore, let
\beq
B(\la,z)=\exp\left(\int_{E_0}^\la \om_{z z^*}\right)
\eeq
be the Blaschke factor on this surface (see e.g.\ \cite {MT} or \cite{tistalg} for more details).
Then, as a corollary of Theorem~\ref{thm:pw} and Theorem~4.3 of \cite{BET} we obtain

\begin{theorem}(Characterization)
Suppose $q(x)$ satisfies \eqref{S.2} and $q(x)=p(x)$ for $x\le a$.

Then a function $R_-(\la)$ is the reflection coefficient for an operator $L_q$
with such a potential if and only if the following conditions are fulfilled:
\begin{itemize}
\item
The function $R_-(\la)$ is continuous on the
set $\siu\cup\sil$ and possess the symmetry property $R_-(\lau)
=\overline{R_-(\lal)}$. Moreover, $|R_-(\la)|<1$ for
$\la\notin\pa\si$ and

\noindent $|R_-(\la)|\leq 1-C|\la-E|$ in a small vicinity of each
point $E\in\pa\si$.

\noindent If $|R_-(E)|=1$, then
\[
R_-(E)=
\begin{cases}
-1 &\text{for } E\notin M_-,\\
1 &\text{for } E\in M_-.
\end{cases}
\]
\item
The function $R_-(\la)\delta_+(\la)$ admits an analytic
continuation to $\C\setminus\{\si\cup\si_d\}$, where
$\si_d=\{\la_1,...,\la_s\}\subset \R\setminus \si $ is a
finite number of real points. Moreover, the function
$g(\la)\delta_-(\la)^{-2} R_-(\la)$ has simple poles at the points
$\la_k$ with
\[
\Res_{\la=\la_k} \frac{g(\la)}{\delta_-(\la)^2}
R_-(\la) >0.
\]
\item
For all large $\la\in\C$ $$\sqrt \la R_-(\la)= O(\E^{2a
\I\sqrt{\la}}).$$
\item
The function $Y^{1/2}(\la)T^{-1}(\la)$, where
\[
T(\la)=\prod_{k=1}^s B^{-1}(\la,\la_k)\exp\left(\frac{1}{2\pi\I}
\oint_\si \log(1-|R_-|^2)w_{\la\la^*}\right),
\]
is continuous up to the boundary $\siu\cup\sil$.
\item
The function
\[
F_{+,c}(x,y)=\oint_\si
\ov{R_-(\la)T^{-1}(\la)}\,T(\la)\psi_+(\la,x)\psi_+(\la,y)g(\la)d\la
\]
satisfies the estimates
\begin{align*}
&|F_{+,c}(x,y)|+\left|\frac{\pa}{\pa x} F_{+,c}
(x,y)\right|\leq Q\left(x+y\right),\\
& \int_0^\infty \left|\frac{d}{dx} F_{+,c}(x,x)\right|(1+ x^2)\,dx<\infty,
\end{align*}
where $Q(x)$ is a continuous, positive, decaying as $x\to +\infty$,
function with $x Q(x)\in L^1(0,\infty)$.
\end{itemize}
\end{theorem}

Note that given a reflection coefficient $R_-(\la)$ as in the previous theorem,
we could form a set of scattering data by choosing arbitrary eigenvalues plus
corresponding norming constants. Then, as long as we take the known algebraic
constraints (see \cite{tistalg}) into account,  we still get a potential $q(x)$ satisfying
\eqref{S.2} by inverse scattering. However, unless \eqref{cond12} holds, this
potential will not satisfy $q(x)=p(x)$.

\section{Applications to KdV}

Finally, we want to show how our main result can be used to prove unique continuation results
for the KdV and MKdV equations.

Let $p(x,t)$ be a real-valued, quasi-periodic, finite-gap solution
of the KdV equation \eqref{KdV} and suppose $q(x,t)$ is a
(classical) solution of \eqref{KdV}, satisfying
\beq\label{conduniq}
\int_{\mathbb R} \big( |q(x,t) - p(x,t)| + |q_t(x,t) - p_{t}(x,t)|
\big) (1+|x|^2)dx <\infty
\eeq
for all $t\in\R$. For the existence
of such solutions we refer to \cite{EGT}, \cite{ET} (see also
\cite{F4}). Then all considerations from the previous section apply
to the operator $L_q(t)$ if we consider $t$ as an additional
parameter. Moreover, the time evolution of the scattering data can
be computed explicitly and is given in the next lemma:

\begin{lemma}[\cite{EGT}]
Let $q(x,t)$ be a solution of the KdV equation satisfying \eqref{conduniq}.
Then $\la_k(t)=\la_k(0)\equiv \la_k;$
\begin{align}\label{refl}
R_\pm(\la,t) &= R_\pm(\la,0)\E^{\alpha_\pm(\la,t)
-\ov{\alpha_\pm(\la,t)}}, \quad \la\in\si, \\ \label{trans} T(\la,t)
&= T(\la,0),\quad\la\in\C,\\ \label{norm} \gamma_k^\pm(t) &=
\gamma_k^\pm(0) \,
\frac{\delta_\pm^2(\la_k,0)}{\delta_\pm^2(\la_k,t)}\,
\E^{2\alpha_\pm(\la_k,t)},
\end{align}
where $\delta_\pm(\la,t)$ is defined  as in \eqref{S2.6} with
$\mu_j^\pm=\mu_j^\pm(t)$, \beq \alpha_\pm(\la,t) := \int_0^t
\left(2(p(0,s) + 2\la) m_\pm(\la,s) - \frac{\pa p(0,s)}{\pa
x}\right)ds \eeq and $m_\pm(\la,t)$ are the Weyl functions of
operator $L_{p}(t)$.
\end{lemma}

Our first result reads

\begin{theorem}\label{thm:kdv1}
Let $p(x,t)$ be a quasi-periodic, finite-gap solution of the KdV equation and
$q(x,t)$ a solution of the KdV equation satisfying \eqref{conduniq}.
Suppose that $q(x,t) = p(x,t)$ for $x<a$ at two times $t_0 \ne t_1$.
Then $q(x,t) = p(x,t)$ for all $(x,t)\in\R^2$.
\end{theorem}

\begin{proof}
Without loss we can choose $t_0=0$. Then $\sqrt\la R_-(\la,0) =
O(\E^{2a\I\sqrt{\la}})$. If $q(.,0) \neq p(.,0)$ we can choose $a$
maximal and this estimate cannot be improved! Thus $\alpha_-(\la,t)=
 -8\I t \la^{3/2} (1+o(1))$ shows that the same estimate cannot hold
for another $t\ne 0$ and we are done.
\end{proof}

This is a special case of a much stronger result from \cite{ekpv}
which states that if $q_1$ and $q_2$ are strong solutions of the KdV
equation such that \beq q_1(\cdot,t_0)-q_2(\cdot,t_0), \:
q_1(\cdot,t_1)-q_2(\cdot,t_1) \in H^1(\R,\E^{a \max(0,x^{3/2})}dx)
\eeq for any $a>0$, then $q_1\equiv q_2$.

With the help of   Theorem \ref{thm:kdv1} we also obtain the
following unique continuation result for our situation:

\begin{theorem}\label{thm:kdv2}
Let $p(x,t)$ be a quasi-periodic, finite-gap solution of the KdV
equation and $q(x,t)$ a solution of \eqref{KdV} satisfying
\eqref{conduniq}. Suppose that $q(x,t) = p(x,t)$ for $(x,t)$ in some
open set $U\subset\R^2$. Then $q(x,t) = p(x,t)$ for all
$(x,t)\in\R^2$.
\end{theorem}

\begin{proof}
Let $[a,b]\times[t_0,t_1] \subset U$ and define
\[
\tilde{q}(x,t)= \begin{cases} p(x,t), & x\leq a,\\
q(x,t), & x \geq a,
\end{cases}
\]
for $t\in[t_0,t_1]$. Then Theorem~\ref{thm:kdv1} implies $\tilde{q}(x,t) = p(x,t)$ for $(x,t)\in\R\times [t_0,t_1]$
and consequently $q(x,t) = p(x,t)$ for $(x,t)\in[a,\infty)\times [t_0,t_1]$. Hence another application of
Theorem~\ref{thm:kdv1} finishes the proof.
\end{proof}

Let $u(x,t)$ be a quasi-periodic, finite-gap solution of the mKdV equation and suppose $v(x,t)$ is
a (classical) solution of the mKdV equation
\beq\label{mKdV}
v_t(x,t) = -v_{xxx}(x,t) + 6 v(x,t)^2 v_x(x,t).
\eeq
Then by virtue of the Miura transform (see, e.g., \cite{GH}, \cite{gs}),
\beq\label{miurau}
p(x,t) = u(x,t)^2 + u_x(x,t),
\eeq
is a quasi-periodic, finite-gap solution of the KdV equation and
\beq\label{miurav}
q(x,t) = v(x,t)^2 + v_x(x,t)
\eeq
is a solution of the KdV equation. We will suppose
again that $q(x,t)$ satisfies \eqref{conduniq} for every $t$.
For the existence of such solutions we refer to \cite{ET3}.

\begin{corollary}
Let $u(x,t)$ be a quasi-periodic, finite-gap solution of the mKdV equation and
$v(x,t)$ a solution of the mKdV equation such that $q(x,t)$ defined by
\eqref{miurav} is a solution of KdV satisfying \eqref{conduniq} with $p(x,t)$ defined by \eqref{miurau}.
Suppose that $v(x,t) = u(x,t)$ for $(x,t)$ in some open set $U\subset\R^2$.
Then $v(x,t) = u(x,t)$ for all $(x,t)\in\R^2$.
\end{corollary}

\begin{proof}
Since $v(x,t) = u(x,t)$ for $(x,t)\in U$ implies $q(x,t) = p(x,t)$ for $(x,t)\in U$, Theorem~\ref{thm:kdv2}
shows $q(x,t) = p(x,t)$ for $(x,t)\in\R^2$. Hence $w_x(x,t)+w(x,t)^2 + 2 u(x,t) w(x,t) =0$, where
$w(x,t) = v(x,t)-u(x,t)$ and the standard uniqueness result for ordinary differential equations
yields $w(x,t)=0$ for $(x,t)\in \R \times \{t | (x_0,t)\in U \text{ for some } x_0\}$. Thus uniqueness
of solutions of the mKdV equation \cite[Thm.~4.1]{ET3} finally implies $w(x,t)=0$ for all $(x,t)\in \R^2$.
\end{proof}

\noindent{\bf Acknowledgments.} We are very grateful to F. Gesztesy for
hints with respect to the literature. G.T. gratefully acknowledges the stimulating atmosphere at the
Centre for Advanced Study at the Norwegian Academy of Science and Letters in Oslo
during June 2009 where parts of this paper were written as part of  the international
research program on Nonlinear Partial Differential Equations.


\begin{thebibliography}{99}
\bibitem{A} T. Aktosun, {\em Bound states and inverse scattering for the Schr\"odinger
equation in one dimension}, J. Math. Phys. {\bf 35:12}, 6231--6236 (1994).
\bibitem{BET} A. Boutet de Monvel, I. Egorova, and G. Teschl,
{\em Inverse scattering theory for one-dimensional Schr\"odinger
operators with steplike finite-gap potentials}, J. d'Analyse Math.
{\bf 106:1}, 271--316, (2008).
\bibitem{dt} P. Deift and E. Trubowitz, {\em Inverse scattering on the line},
Commun. Pure Appl. Math. {\bf 32}, 121--251 (1979).
\bibitem{ET} I. Egorova and G. Teschl, {\em On the Cauchy
problem for the Korteweg--de Vries equation with steplike finite-gap initial data II.
Perturbations with Finite Moments}, \arxiv{0909.1576}.
\bibitem{ET3} I. Egorova and G. Teschl, {\em On the Cauchy
problem for the modified Korteweg--de Vries equation with steplike finite-gap initial data},
\arxiv{0909.3499}.
\bibitem{EGT} I. Egorova, K. Grunert, and G. Teschl, {\em On the Cauchy
problem for the Korteweg--de Vries equation with steplike finite-gap initial data I.
Schwartz-type perturbations}, Nonlinearity {\bf 22}, 1431--1457 (2009).
\bibitem{ekpv} L. Escauriaza, C. E.  Kenig, G. Ponce, and L. Vega, {\em On uniqueness properties of
solutions of the k-generalized KdV equations}, J. Funct. Anal. {\bf 244}, 504--535 (2007).
\bibitem{F1} N.E. Firsova, {\it An inverse scattering problem for the perturbed Hill
operator},  Mat. Zametki {\bf 18}, no. 6, 831--843 (1975).
\bibitem{F2} N.E. Firsova, {\it A direct and inverse scattering problem for a one-dimensional
perturbed Hill operator} Matem. Sborn. (N.S.) {\bf 130(172)}, no. 3, 349--385 (1986).
\bibitem{F3} N. Firsova, {\em Resonances of the perturbed Hill operator with exponentially decreasing
extrinsic potential}, Mat. Zametki {\bf 36}, 711--724 (1984).
\bibitem{F4} N.E. Firsova, {\em Solution of the Cauchy problem for the Korteweg-de Vries equation
with initial data that are the sum of a periodic and a rapidly decreasing function},
Math. USSR-Sb. {\bf 63}, no. 1, 257--265 (1989).
\bibitem{GGKM} C. S. Gardner, J. M. Green, M. D. Kruskal, R. M. Miura,
 {\it Method for solving the Korteweg-de Vries equation},
Phys. Rev. Lett., {\bf 19}, 1095--1097 (1967).
\bibitem{GH} F. Gesztesy and H. Holden, {\em Soliton Equations and
their Algebro-Geometric Solutions. Volume {I}: $(1+1)$-Dimensional
Continuous Models}, Cambridge Studies in Advanced Mathematics,
Vol. {\bf 79}, Cambridge University Press, Cambridge, 2003.
\bibitem{gs} F. Gesztesy and R. Svirsky, {\em (m)KdV-Solitons on the background of
quasi-periodic finite-gap solutions}, Memoirs Amer. Math. Soc. {\bf 118},
No. 563 (1995).
\bibitem{gt} F. Gesztesy and G. Teschl, {\em On the double commutation
method}, Proc. Amer. Math. Soc. {\bf 124}, 1831--1840 (1996).
\bibitem{GRT} F. Gesztesy, R. Ratnaseelan, and G. Teschl, {\it The KdV
hierarchy and associated trace formulas}, in ``{\it Proceedings of
the International Conference on  Applications of Operator Theory\/}'',
(eds. I. Gohberg, P. Lancaster, and P. N. Shivakumar), Oper. Theory
Adv. Appl., 87, Birkh\"auser, Basel, 125--163 (1996).
\bibitem{kt} H. Kr\"uger and G. Teschl, {\em Unique continuation for discrete nonlinear wave equations},
\arxiv{0904.0011}.
\bibitem{M} V. A. Marchenko, {\em Sturm--Liouville Operators and
Applications}, Birkh\"auser, Basel, 1986.
\bibitem{MT} A. Mikikits-Leitner and G. Teschl, {\em Trace formulas for
Schr\"odinger operators in connection with scattering theory for finite-gap backgrounds},
\arxiv{0902.3917}.
 \bibitem{NMPZ} S. P. Novikov, S. V. Manakov, L. P. Pitaevskii, and V. E. Zakharov,
{\it Theory of Solitons. The Inverse Scattering Method}, Springer, Berlin, 1984.
\bibitem{ry} A. Rybkin, {\em Meromorphic solutions to the KdV equation with
non-decaying initial data supported on a left half-line}, Preprint.
\bibitem{ta} S. Tarama, {\em Analytic solutions of the Korteweg--de Vries equation}, J. Math. Kyoto Univ. {\bf 44}, 1--32 (2004).
\bibitem{tistalg} G. Teschl, {\em Algebro-geometric constraints on solitons with respect to
quasi-periodic backgrounds}, Bull. London Math. Soc. {\bf 39}, No.4, 677--684 (2007).
\bibitem{zh} B. Zhang, {\em Unique continuation for the Korteweg--de Vries equation}, SIAM J. Math. Anal. {\bf 23},
55--71 (1992).
\end{thebibliography}
\end{document}